\documentclass[journal]{IEEEtran}
\usepackage{amsmath,amsfonts,amssymb,amsthm}
\usepackage{textcomp}
\usepackage{stfloats}
\usepackage{url}
\usepackage{verbatim}
\usepackage{graphicx}    % Always load graphicx first.
\usepackage{epstopdf}    % This triggers EPS-->PDF conversion.
\usepackage{pgfplots}
\pgfplotsset{compat=1.18}
\usepackage{balance}
\usepackage{mathrsfs}
\usepackage{cite}
\usepackage{xcolor}
\usepackage[bookmarks=true,colorlinks=true,allcolors=blue,allbordercolors=blue]{hyperref}
\usepackage{multirow}
\usepackage{bm}
\usepackage{comment}

\usepackage{array}
\usepackage{booktabs}
\usepackage{pifont}  % for \ding{55}

\usepackage[table]{xcolor}
\usepackage{colortbl}
\definecolor{lightblue}{RGB}{220,230,245}

\usepackage{color}
\definecolor{amaranth}{rgb}{0.9, 0.17, 0.31}
\definecolor{AlirezaPurple}{RGB}{150, 0, 250}

\usepackage{subcaption}

\usepackage[font=small]{caption}
\usepackage{pdfpages}

\usepackage{tikz}
\usetikzlibrary{tikzmark,decorations.pathreplacing, patterns}
\usetikzlibrary{arrows, arrows.meta}

\usepackage{rotating}

\usepackage{algorithm}
\usepackage{algorithmic}

\usepackage{fancyhdr}

\usepackage{acronym}

\usepackage{authblk}

\usepackage{soul}

\usepackage{amsthm}
\usepackage{babel}

\theoremstyle{plain}

\newtheorem{prop}{\protect\propname}

\theoremstyle{remark}

\providecommand{\theoremname}{Theorem}
\providecommand{\propname}{Proposition}
\providecommand{\remarkname}{Remark}

\providecommand{\corollaryname}{Corollary}

\usepackage[super]{nth}
\allowdisplaybreaks

\usepackage{acronym}

\acrodef{2d}[2D]{bi-dimensional}
\acrodef{5g}[5G]{fifth-generation}
\acrodef{6g}[6G]{sixth-generation}
\acrodef{aoa}[AoA]{angle-of-arrival}
\acrodef{ao}[AO]{alternating optimization}
\acrodef{aod}[AoD]{angle-of-departure}
\acrodef{bs}[BS]{base station}
\acrodef{csi}[CSI]{channel state information}
\acrodef{cp}[CP]{canonical polyadic}
\acrodef{cfar}[CFAR]{constant false alarm rate}
\acrodef{ca-cfar}[CA-CFAR]{cell averaging \ac{cfar}}
\acrodef{ccrb}[CCRB]{constrained Cramér-Rao bound}
\acrodef{crb}[CRB]{Cramér-Rao bound}
\acrodef{dais}[DAIS]{delay-angle information spoofing}
\acrodef{dft}[DFT]{discrete Fourier transform}
\acrodef{esprit}[ESPRIT]{estimation of signal parameters via rotational invariance techniques}
\acrodef{fft}[FFT]{fast Fourier transform}
\acrodef{fim}[FIM]{Fisher information matrix}
\acrodef{fr}[FR]{frequency range}

\acrodef{gnss}[GNSS]{global navigation satellite system}
\acrodef{gospa}[GOSPA]{generalized optimal sub-pattern assignment}
\acrodef{hd}[HD]{high-definition}
\acrodef{los}[LoS]{line-of-sight}
\acrodef{ls}[LS]{least-squares}

\acrodef{mcrb}[MCRB]{mismatched \ac{crb}}
\acrodef{mmWave}[mmWave]{Millimeter-wave}
\acrodef{mimo}[MIMO]{multiple-input multiple-output}
\acrodef{miso}[MISO]{multiple-input single-output}
\acrodef{mpc}[MPC]{multipath component}
\acrodef{ml}[ML]{maximum likelihood}
\acrodef{mf}[MF]{matched-filter}
\acrodef{nlos}[NLoS]{non-line-of-sight}
\acrodef{ofdm}[OFDM]{orthogonal frequency division multiplexing}
\acrodef{omp}[OMP]{orthogonal matching pursuit}
\acrodef{peb}[PEB]{position error bound}
\acrodef{pla}[PLA]{physical layer authentication}
\acrodef{poam}[POAM]{penalized optimal assignment metric}
\acrodef{prmse}[PRMSE]{penalized root mean squared error}
\acrodef{prs}[PRS]{positioning reference signal}
\acrodef{roi}[ROI]{region of interest}
\acrodef{rf}[RF]{radio frequency}
\acrodef{rfc}[RFC]{radio frequency chain}
\acrodef{rcs}[RCS]{radar cross section}
\acrodef{ris}[RIS]{reflective intelligent surface} 
\acrodef{rmse}[RMSE]{root mean square error}
\acrodef{simo}[SIMO]{single-input multiple-output}
\acrodef{siso}[SISO]{single-input single-output}
\acrodef{snr}[SNR]{signal-to-noise ratio}
\acrodef{sota}[SoTA]{state-of-the-art}
\acrodef{sp}[SP]{scatter point}
\acrodef{svd}[SVD]{singular value decomposition}
\acrodef{tdoa}[TDoA]{time difference of arrival}
\acrodef{toa}[ToA]{time of arrival}
\acrodef{qcqp}[QCQP]{quadratically constrained quadratic program}

\acrodef{ue}[UE]{user equipment} 
\acrodef{ula}[ULA]{uniform linear array}
\acrodef{ura}[URA]{uniform rectangular array}
\acrodef{va}[VA]{virtual anchor}

\newtheorem{remark}{Remark}

\begin{document}
\bstctlcite{IEEEexample:BSTcontrol}
% \title{A Blind Analog Precoder Design for Location Privacy in MIMO Systems
% }
\title{From Pilot to Precoding Design: Blind Angular Spoofing For Location Privacy in MIMO Systems}


\author{
\IEEEauthorblockN{Priyanka Maity\IEEEauthorrefmark{1}, Lorenzo Italiano\IEEEauthorrefmark{2}, Alireza Pourafzal\IEEEauthorrefmark{1}, Gonzalo Seco-Granados\IEEEauthorrefmark{3},}

\IEEEauthorblockN{Hui Chen\IEEEauthorrefmark{1}, Monica Nicoli\IEEEauthorrefmark{2}, Henk Wymeersch\IEEEauthorrefmark{1}}

\IEEEauthorblockA{\IEEEauthorrefmark{1}Chalmers University of Technology, G\"oteborg, Sweden}

\IEEEauthorblockA{\IEEEauthorrefmark{2}Politecnico di Milano, Milan, Italy}

\IEEEauthorblockA{\IEEEauthorrefmark{3}Autonomous University of Barcelona, Barcelona, Spain}\vspace{-20pt}
\thanks{This work has been supported, in part, by the SNS JU project 6G-DISAC under the Grant Agreement No 101139130 and in part by Vinnova under project number 2026-00976.}}

% <-this % stops a space
%\thanks{M. Shell was with the Department
%of Electrical and Computer Engineering, Georgia Institute of Technology, Atlanta,
%GA, 30332 USA e-mail: (see %http://www.michaelshell.org/contact.html).}% <-this % stops a space

 \maketitle

\begin{abstract}
This paper studies location privacy in uplink MIMO systems, where a user equipment seeks to spoof the angular signature observed by a single base station  performing localization. We propose a blind analog precoder design that manipulates the perceived  angle-of-arrival  and angle-of-departure  configuration without requiring channel-gain knowledge. The method enforces consistency between the received signal and a desired spoofed angular subspace, and is solved using an alternating optimization algorithm under practical amplitude constraints. Simulations in a multipath scenario show that the proposed approach achieves near-perfect angular spoofing and clearly outperforms pilot-only blind spoofing, which exhibits an error floor. The results also show a trade-off between spoofing accuracy and communication rate, depending on the chosen virtual geometry.  
%This paper studies location privacy in uplink MIMO systems, where a user equipment (UE) aims to manipulate the spatial signature observed at a single base station (BS) performing localization from multipath channel measurements. In contrast to prior approaches, we propose a blind analog precoder design that enables controlled spoofing of the angle-of-arrival (AoA) and angle-of-departure (AoD) signatures without requiring knowledge of the channel gains. The problem is formulated by designing time-varying precoders such that the received signal lies in the observation subspace associated with a desired set of target angles, thereby making the observed angular signature consistent with a spoofed channel. An alternating optimization algorithm is developed to jointly update the precoder and auxiliary variables under practical amplitude constraints. Simulation results in a multipath scenario show that the proposed method achieves near-perfect spoofing of the target angular signature, while pilot-only blind spoofing exhibits a clear error floor. The results further reveal a trade-off between spoofing accuracy and communication performance, where highly effective manipulation of angular cues used for positioning can be achieved, but the achievable rate depends on how well the chosen virtual configuration aligns with the underlying channel conditions.
\end{abstract}


\section{Introduction}

In modern cellular systems, large antenna arrays and wide bandwidths enable the extraction of channel features such as \ac{toa}, \ac{aoa}, and \ac{aod}, which together enable increasingly accurate position estimation \cite{witrisal2016high, wymeersch20185g, italiano2025tutorial}. These quantities are directly tied to the underlying  geometry, since they describe the delay and directions of the paths linking the transmitter, receiver, and scatterers \cite{Shahmansoori2018}. Such geometric features serve as the basis for downstream spatial inference and can be combined to estimate user position \cite{Sun2022,Lin2020}. This raises significant privacy concerns, especially when location-relevant information can be inferred without the user’s explicit consent \cite{abedi2022non}.

\begin{figure}
    \centering
    \includegraphics[width=0.9\linewidth]{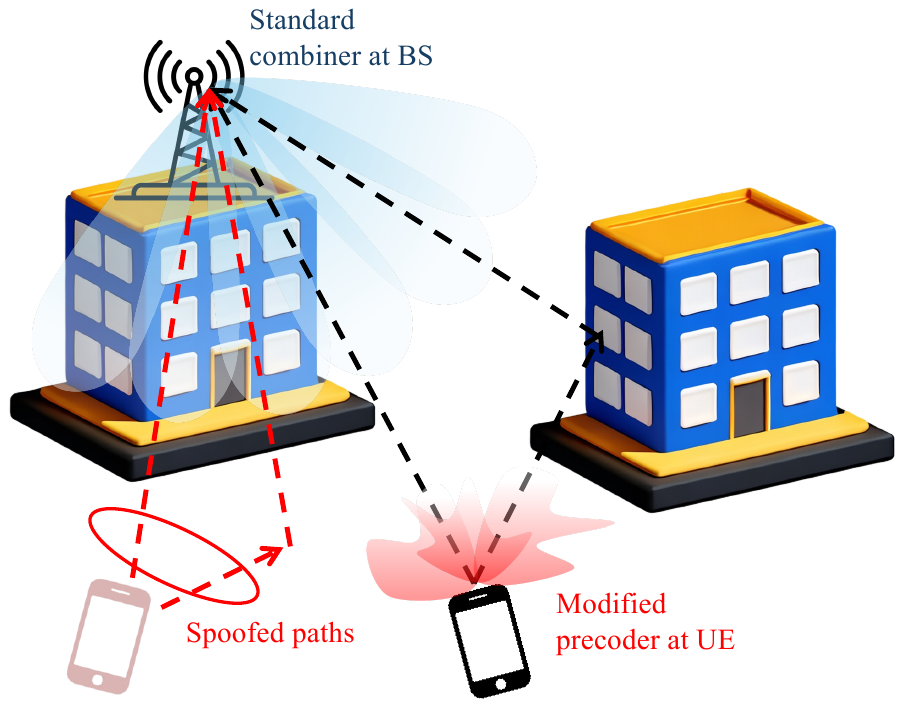}
    \caption{By modifying its uplink precoder, the UE shapes the received signal at the BS so that spoofed paths (red) are perceived instead of the true propagation paths (black), leading to a manipulated angular signature at the receiver.}
    \label{fig:model}
\end{figure}

%One way to protect location privacy is to manipulate the transmitted signal so that the receiver infers a distorted or misleading geometric signature \cite{checa2020location}. The effectiveness of such an approach depends on both the information available at the \ac{ue} about the channel and its ability to shape the transmitted signal. Existing works have considered privacy enhancement under full or partial \ac{csi} at the \ac{ue}, using this knowledge to optimize the waveform or beamforming strategy \cite{ardagna2007location,  zhang2024privacy, tomasin2022beamforming, zhang_privacy_2025, khan2025beamforming}. Such \ac{csi} may be obtained through channel estimation, reciprocity, or prior environmental knowledge. However, this assumption is restrictive in privacy-sensitive scenarios, where the \ac{ue} may not have reliable knowledge of the channel toward an unauthorized or potentially malicious \ac{bs}. This motivates the study of \ac{csi}-free, or \emph{blind}, privacy mechanisms \cite{li2024channel, li2025delay, italiano2025holotrace}, which aim to influence localization-relevant features without relying on explicit channel knowledge. 

One way to protect location privacy is to manipulate the transmitted signal so that the receiver observes a distorted geometric signature rather than the true one. Early work explored this idea through location obfuscation in a broad privacy context \cite{ardagna2007location}, while more recent wireless-specific studies showed that signal design can directly mislead mmWave localization mechanisms \cite{checa2020location}. Building on this principle, several works have assumed that the \ac{ue} has at least partial \ac{csi} and can therefore optimize its transmission accordingly. For example, privacy in delay-based localization has been enhanced through artificial noise or artificial multipath design \cite{zhang2024privacy}, and beamforming-based privacy mechanisms have been developed for cellular and MIMO-OFDM systems \cite{tomasin2022beamforming, zhang_privacy_2025, khan2025beamforming}. These approaches rely on channel knowledge obtained through estimation, reciprocity, or prior knowledge of the environment. In privacy-sensitive settings, however, such assumptions may be unrealistic, especially when the \ac{bs} is unauthorized or potentially malicious. This has motivated the development of \textcolor{black}{geometry-aware} \ac{csi}-blind privacy mechanisms, including fake-path injection \cite{li2024channel}, delay-angle spoofing without explicit \ac{csi} \cite{li2025delay}, and privacy preservation for mmWave MIMO-OFDM systems \cite{italiano2025holotrace}.

Blind privacy mechanisms still offer limited control over the geometry perceived by the receiver. The method in \cite{li2025delay} perturbs delay-based localization in a fully digital architecture, but it does not provide direct control over the angular structure observed at the \ac{bs}. The pilot-based blind spoofing approach in \cite{italiano2025holotrace}  avoids explicit \ac{csi}, but it likewise lacks a structured mechanism to enforce a desired spoofed angular configuration. This limitation is important because, as observed in \cite{italiano2025holotrace}, the chosen spoofed geometry can strongly affect the resulting communication performance. These limitations are particularly relevant when the \ac{ue} is restricted to analog precoding. In that case, the available design freedom is primarily spatial, which makes angular features the natural target for manipulation.

In this paper, \textcolor{black}{we address this lack of direct control over the angular geometry perceived by the receiver} by studying angular-signature spoofing in a single \ac{bs} narrowband multipath system, where the \ac{ue} is constrained to analog precoding and the available design freedom is therefore primarily spatial. Accordingly, we focus on the manipulation of the perceived \ac{aoa}/\ac{aod} configuration and deliberately isolate the angular component of the localization problem. Specifically, we design time-varying analog precoders that make the received uplink observation consistent with a desired false angular geometry without requiring knowledge of the channel gains. Our contributions are threefold: (i) we formulate blind angular spoofing as the design of analog precoders that force the received signal into the observation subspace associated with a target set of angles; (ii) we develop an alternating optimization method to solve the resulting constrained problem under practical amplitude constraints; and (iii) we show numerically that the proposed method achieves substantially more accurate angular spoofing than pilot-only blind methods, while also revealing the trade-off between spoofing performance and achievable communication rate.
%\vspace{-5mm}
\section{System Model}
%\subsection{{Scenario Description}}
We consider an uplink narrowband system with an $N_r$-antenna BS and an $N_t$-antenna UE, as illustrated in Fig. \ref{fig:model}. The BS is located at $\mathbf{p}_{\text{BS}}=[0,0]^T$ and orientation $o_\mathrm{BS}^{}$, and receives signals
from a UE with unknown position $\mathbf{p}_{\text{UE}}$ and orientation $o_\mathrm{UE}^{}$. The signal propagates through $L$ paths: one LoS path $(l =
0)$  and $L-1$ non-line-of-sight (NLoS) paths. We consider a narrowband system in which the channel is frequency-flat, and delay differences are not resolvable.
%\subsection{{Channel and Signal Model}}
The transmit pilot at symbol $m \in \{1,\dots,M\}$, and measurement $s~\in~\{1,\dots,S\}$ is denoted by $x_{s,m}^{}~\in~\mathbb{C}$. The pilot symbols are chosen as $x_{s,m}^{}=1$ for all $s,m$.
The exhaustive beamforming approach is employed \cite{beamforming, 6847111}: the transmit precoder $\mathbf{f}_{m}~\in~\mathbb{C}^{N_t^{}\times 1}$ changes in each $m$-th symbol. The combiner at the BS $\mathbf{w}_{s}~\in~\mathbb{C}^{N_r^{}\times 1}$  is kept constant over $M$ symbols to complete one measurement index $s$.  The channel is assumed to remain constant across all $MS$ symbols or measurement indices.
Assuming that the total transmit power is $P_t^{}$, the received signal at symbol $m$, and measurement $s$ is
\begin{align}
    y_{s,m}&= \sqrt{P_t}\sum_{l=0}^{L-1}\alpha_l\mathbf{w}_{s}^\mathsf{H} \mathbf{a}_{\mathrm{BS}}(\theta_l) \mathbf{a}_{\mathrm{UE}}^\top(\phi_l)  \mathbf{f}_{m}  x_{s,m} + n_{s,m}  \\
&= \sqrt{P_t}\mathbf{w}_{s}^\mathsf{H}  \mathbf{A}_{\mathrm{BS}} \boldsymbol{\Lambda}\mathbf{A}_{\mathrm{UE}}^\top {\mathbf{f}}_{m}  + n_{s,m}   \\
&= \sqrt{P_t}\mathbf{w}_{s}^\mathsf{H} \mathbf{H}  {\mathbf{f}}_{m} + n_{s,m}   , \label{eq:rx_signal}
\end{align}
where $\alpha_l$ is the complex path gain, $\theta_l$ and $\phi_l$ are the AoA/AoD,  $\bm{\Lambda} = \operatorname{diag}(\bm{\alpha})$, $\mathbf{A}_{\mathrm{BS}} ~=~[\mathbf{a}_{\mathrm{BS}} (\theta_0^{}), \cdots, \mathbf{a}_{\mathrm{BS}} (\theta_{L-1}^{})]$ and $\mathbf{A}_{\mathrm{UE}} ~=~[\mathbf{a}_{\mathrm{UE}} (\phi_0^{}), \cdots , \mathbf{a}_{\mathrm{UE}} (\phi_{L-1}^{})]$ are the array response matrices at BS and UE. For a half-wavelength spaced uniform linear array (ULA) with $N$ antennas, the array response at angle $\vartheta$ is $\mathbf{a}(\vartheta) = \big[1,e^{j\pi \sin(\vartheta)}, \cdots, e^{j\pi(N-1) \sin(\vartheta)}\big]^T$. The noise samples are modeled as { $n_{s,m}^{} \sim \mathcal{CN}(0,BN_0^{})$ where $B$ is available bandwidth and $N_0^{}$ the noise power spectral density.} The LoS measurements $\left\{ \theta_0, \phi_0\right\}$ are related to the UE state by  $\theta_0^{}  = \text{atan2}(y_{\text{BS}}^{} - y_{\text{UE}}^{},\; x_{\text{BS}}^{} - x_{\text{UE}}^{}) - o_{\text{BS}}^{}$ and $\phi_0^{}  = \text{atan2}(y_{\text{UE}}^{} - y_{\text{BS}}^{},\; x_{\text{UE}}^{} - x_{\text{BS}}^{}) - o_{\text{UE}}^{}$.
% It is important to note that $o_{\text{BS}}^{}$ and $o_{\text{UE}}^{}$ are unknowns in the positioning problem. 
Stacking the received signals $\{y_{s,m}\}$ into a matrix $\mathbf{Y} \in \mathbb{C}^{S \times M}$ with $[\mathbf{Y}]_{s,m} = y_{s,m}$, we obtain
\begin{equation}
\mathbf{Y} = \sqrt{P_t}\,\mathbf{W}^\mathsf{H} \mathbf{H} \mathbf{F} + \mathbf{N}, \label{Y_compact}
\end{equation}
where $\mathbf{W} = [\mathbf{w}_1,\dots,\mathbf{w}_S]$ and $\mathbf{F} = [\mathbf{f}_1,\dots,\mathbf{f}_M]$.

\section{Problem Formulation}
In this section, we describe the system operation, including both the channel estimation and communication phases, followed by the adversarial threat model.
\subsection{Operational Model at BS}
We consider a two-phase operation at the BS, consisting of a channel estimation phase followed by a communication phase. The UE is assumed to transmit pilot symbols using a precoder sequence $\mathbf{f}_m$. %; under adversarial operation, this may be replaced by adaptive precoders $\widetilde{\mathbf{f}}_{s,m}$, 
%varying across measurements and symbols.
Vectorizing  $\mathbf{Y}$ in \eqref{Y_compact} yields
%Assuming $\mathbf{x}_{m} = \mathbf{1}$, 
\begin{align}
    \mathbf{y} &= \sqrt{P_t}({\mathbf{F}}^\top \mathbf{A}_{\mathrm{UE}} \circledast \mathbf{W}^\mathsf{H}  \mathbf{A}_{\mathrm{BS}}) \boldsymbol{\alpha} +  \mathbf{n} 
    \\ &= \sqrt{P_t}\mathbf{Z} \boldsymbol{\alpha} +  \mathbf{n}\in~\mathbb{C}^{SM\times 1}.\label{eq:nominalBSmodel}
\end{align}
where $\circledast$ denotes the Khatri-Rao product.
% Stacking the received signals over $S$ measurements, 
% \begin{align}
%         \mathbf{y} = \mathbf{Z} \boldsymbol{\alpha} +  \mathbf{n} \in~\mathbb{C}^{SM\times 1},
% \end{align}
The BS  estimates the channel parameters as 
\begin{equation}
 (\widehat{\bm{\theta}}, \widehat{\bm{\phi}},\widehat{\bm{\alpha}}) = \underset{ \bm{\theta},\bm{\phi}, \bm{\alpha} }{\text{arg\,min}} \, \Vert \mathbf{y}  - \sqrt{P_t}\mathbf{Z}\boldsymbol{\alpha}\Vert^{2}. 
\end{equation}
The vectors $\bm{\theta},\bm{\phi}$ collect the AoAs and AoDs of the $L$ propagation paths, respectively and define $\bm \rho \triangleq [\bm{\theta},\bm{\phi}]$ as the angular parameter vector.
The channel gain estimate $\widehat{\bm{\alpha}}$ is given by the least squares solution as $\widehat{\bm{\alpha}} = (1/\sqrt{P_t}) \mathbf{Z}^\dagger \mathbf{y} $, where $^\dagger$ denotes Moore-Penrose inverse. The projected cost is then given as
\begin{equation}
    C(\bm{\rho}) = \left\| \left(\mathbf{I}_{SM}^{} - \mathbf{Z}\mathbf{Z}^\dagger\right)\mathbf{y} \right\|^2. \label{cost_function}
\end{equation}
Using a grid-based maximum likelihood approach, the BS evaluates the likelihood over a discretized \ac{aoa}–\ac{aod} grid and identifies the $L$ peaks to obtain the estimates $\widehat{\bm{\theta}}, \widehat{\bm{\phi}}$. The BS estimates the geometric channel parameters, namely the AoAs and AoDs, from the received signal. These estimated channel parameters are subsequently used for communication beam selection: 
the \ac{bs} selects the beam pair maximizing total received power, i.e.,
$(\hat s,\hat m)=\arg\max_{s',\,m'} \lvert y_{s',m'}\rvert^2$
%\begin{equation}
%    $\widehat{m}, \widehat{s} =  \underset{ m', s' }{\text{arg\,max}} |y_{m',s'}^{}|^2$, %\quad 
%\end{equation}
where $\mathbf{w}_\mathrm{comm}^{} = \mathbf{w}_{\widehat{s}}^{},$ and $ \mathbf{f}_\mathrm{comm}^{} = \mathbf{f}_{\widehat{m}}^{}.$ 
The achievable sum rate is  $R = B \log_2 (1 + {\gamma^{} P_t^{}}/{N_0^{}} )$, 
%\begin{equation}
 %   R = B \log_2 \left(1 + \frac{\gamma^{} P_t^{}}{N_0^{}} \right),
  %  \label{eq:comm_rate_nospoof}
%\end{equation}
% with $\gamma_k = \left| \sum_{\ell=0}^{L-1} \alpha_\ell\, \mathbf{w}_\mathrm{comm}^{\mathsf{T}} \mathbf{a}_\mathrm{BS}(\theta_\ell)\, \mathbf{a}_\mathrm{UE}^{\mathsf{T}}(\phi_\ell)\, \mathbf{f}_\mathrm{comm}\, e^{-j2\pi k\Delta f \tau_\ell} \right|^2$.
with $\gamma^{}=|\sum_{\ell=0}^{L-1}\alpha_{\ell}^{}\mathbf{w}_{\text{comm}}^{\mathsf{H}}\mathbf{a}_{\text{BS}}^{}(\theta_{\ell}^{})\mathbf{a}_{\text{UE}}^{\mathsf{T}}(\phi_{\ell}^{})\mathbf{f}_{\text{comm}}|^{2}$. 

\subsection{Adversarial UE Model and Problem Statement}
We consider an adversarial UE that replaces the nominal fixed precoder $\mathbf F$ with adaptive precoders $\widetilde{\mathbf f}_{s,m}$ varying across the measurement index $s$ and symbol index $m$. The objective is to make the BS, which processes the observations according to \eqref{eq:nominalBSmodel}, interpret the received signal as arising from a desired false angular geometry. Hence, the problem is to design $\widetilde{\mathbf f}_{s,m}$ such that the received uplink signal is consistent with the observation subspace of target AoA/AoD vectors (which we denote by $\overline{\bm{\rho}} = (\overline{\bm{\theta}},\overline{\bm{\phi}}) $), while satisfying the analog-precoding constraints and remaining temporally consistent during communication.

\section{Methodology}
In this section, we mathematically formalize the precoding design problem and propose an alternating optimization strategy. 
\subsection{Joint AoA/AoD Spoofing Formulation}
The UE designs precoders $\{\widetilde{\mathbf{f}}_{s,m}\}$ which results in the spoofed signal as follows
\begin{align}
    y_{s,m}^{\text{spoof}}&=  \sqrt{P_t}\mathbf{w}_{s}^\mathsf{H} \mathbf{H}  {\widetilde{\mathbf{f}}}_{s,m} + n_{s,m}   . \label{eq:rx_signal_spoofed}
\end{align}
 Concatenating the received signals over $M$ symbols and vectorizing, one obtains
    \begin{align}
% \mathbf{y}_{s}^\top &= \sqrt{P_t}\mathbf{w}_{s}^\mathsf{H}  \mathbf{H}\widetilde{\mathbf{F}}_{s}  + \mathbf{n}_{s}^\top \nonumber \\
 \mathbf{y}_{s} &= \sqrt{P_t}(\mathbf{w}_s^\mathsf{H}  \mathbf{A}_{\mathrm{BS}}\circledast \widetilde{\mathbf{F}}_s^\top \mathbf{A}_{\mathrm{UE}}) \boldsymbol{\alpha}  + \mathbf{n}_{s} \nonumber \\ &= \sqrt{P_t}\mathbf{Z}_s(\widetilde{\mathbf{F}}_s,\bm{\rho})  \boldsymbol{\alpha}  + \mathbf{n}_{s} \label{y_tensor_spoof},
\end{align}
where $ \widetilde{\mathbf{F}}_{s}= [\widetilde{\mathbf{f}}_{s,1} , \widetilde{\mathbf{f}}_{s,2} , \cdots, \widetilde{\mathbf{f}}_{s,M} ]$. Stacking across all the $S$ measurements, the received signal is given as $ \mathbf{y} = \sqrt{P_t}\mathbf{Z}\left(\{\widetilde{\mathbf{F}}_s\}, \bm{\rho}\right)  \boldsymbol{\alpha}  + \mathbf{n} \in~\mathbb{C}^{SM\times 1}$, where $\mathbf{Z}\left(\{\widetilde{\mathbf{F}}_s\}, \bm{\rho}\right)  = [\mathbf{Z}^{\top}_1(\widetilde{\mathbf{F}}_s, \bm{\rho}),\mathbf{Z}^{\top}_2(\widetilde{\mathbf{F}}_s, \bm{\rho}),\cdots,\mathbf{Z}^{\top}_S(\widetilde{\mathbf{F}}_s, \bm{\rho})]^\top \in~\mathbb{C}^{SM\times L}$. The BS processes the received signal generated by the spoofing precoders $\{\widetilde{\mathbf{F}}_s\}$ and evaluates the cost function $C(\bm{\rho}) $ with the nominal precoders $\mathbf{F}$, as defined in \eqref{cost_function}.  Therefore, spoofing requires that the received signal generated by $\{\widetilde{\mathbf{F}}_s\}$ is consistent with the nominal observation model evaluated at the target angles $\overline{\bm{\rho}} $, i.e., 
\begin{equation}
C({\bm{\rho}}; \{\widetilde{\mathbf{F}}_s\}) \!= \!\! \left\| \left(\mathbf{I}_{SM}^{} - \mathbf{Z}(\mathbf{F}, {\bm{\rho}}) \mathbf{Z}^\dagger(\mathbf{F}, {\bm{\rho}})\right) \mathbf{y}  \right\|^2 \label{eq:spoofingCriterion1}
\end{equation}
is minimal at $ \overline{\bm{\rho}} $.
 % Here $\bm{\lambda} \in \mathbb{C}^{L \times 1}$ is the spoofed channel gain. The quantities $\mathbf{Z}(\overline{\bm{\theta}},\overline{\bm{\phi}}) , \mathbf{Z}({\bm{\theta}},{\bm{\phi}}) $ are given as $\mathbf{Z}_s(\overline{\bm{\theta}},\overline{\bm{\phi}}) = \mathbf{w}_{s}^\top  \mathbf{A}_{\mathrm{BS}}(\overline{\bm{\theta}}) \circledast \overline{\mathbf{F}}_s^\top \mathbf{A}_{\mathrm{UE}}(\overline{\bm{\phi}})$, and $\mathbf{Z}_s({\bm{\theta}},{\bm{\phi}}) = \mathbf{w}_{s}^\top  \mathbf{A}_{\mathrm{BS}}({\bm{\theta}}) \circledast \widetilde{\mathbf{F}}_s^\top \mathbf{A}_{\mathrm{UE}}({\bm{\phi}})$. 
 A perfect spoofing attack is achieved if, for any channel gain vector $\bm{\alpha}$, there exists a vector $\bm{\lambda}$, such that the received signal induced by the spoofing precoders and true channel parameters $(\boldsymbol{\rho}, \boldsymbol{\alpha})$ is indistinguishable from that of a target channel $(\bar{\boldsymbol{\rho}}, \boldsymbol{\lambda})$ with nominal precoders. Mathematically, this corresponds to
%\begin{equation}
$\mathbf{Z}(\{\widetilde{\mathbf{F}}_s\}, \boldsymbol{\rho}) \boldsymbol{\alpha}
=
\mathbf{Z}(\mathbf{F}, \overline{\boldsymbol{\rho}}) \boldsymbol{\lambda}$.
%\end{equation} 
Here, $\bm{\lambda} \in \mathbb{C}^{L \times 1}$ denotes an equivalent channel gain associated with the target angles.
% To spoof the estimate toward the target angles $(\overline{\bm{\theta}}, \overline{\bm{\phi}})$, the estimator cost at the candidate $(\overline{\bm{\theta}}, \overline{\bm{\phi}})$ is minimized. In particular, perfect spoofing (as defined above) is achieved when this cost becomes zero, i.e.,
% \begin{equation}
% C(\overline{\bm{\theta}},\overline{\bm{\phi}}; \widetilde{\mathbf{F}}_s) \!= \!\! \left\| \left(\mathbf{I}_{M}^{} - \mathbf{Z}_s(\overline{\bm{\theta}},\overline{\bm{\phi}}) \mathbf{Z}_s^\dagger(\overline{\bm{\theta}},\overline{\bm{\phi}})\right) \mathbf{Z}_s(\bm{\theta},\bm{\phi})\, \bm{\alpha}  \right\|^2. \label{eq:mimo_cost}
% \end{equation}
\begin{prop}\label{prop:BlindSpoof}
    To achieve perfect spoofing of a target \ac{aoa} vector $\overline{\bm{\theta}}$ and \ac{aod} vector $\overline{\bm{\phi}}$  for estimator \eqref{eq:spoofingCriterion1},
    % under the constraint $\bm{\lambda} \in \mathcal{C}$ with $ \mathcal{C} = \{\bm{\lambda} : \big\| \bm{\lambda} \big\|^2  \leq P_{\text{max}} \}$,
    the cost function $C(\overline{\bm{\rho}}; \{\widetilde{\mathbf{F}}_s\}) = 0 $ iff
    % $   \forall \bm{\alpha}, \exists \bm{\lambda} \in \mathcal{C}$ such that,
    %     \begin{equation}
    %  \mathbf{Z}_s(\bm{\theta},\bm{\phi})\bm{\alpha} = \mathbf{Z}_s(\overline{\bm{\theta}},\overline{\bm{\phi}}) \bm{\lambda}  \label{theorem_eq}
    % \end{equation}
    % \begin{equation}
    %     \mathcal C(\mathbf{Z}_s(\bm{\theta},\bm{\phi})) \subseteq \mathcal C(\mathbf{Z}_s(\overline{\bm{\theta}},\overline{\bm{\phi}}) ) \label{theorem_eq}
    % \end{equation}
    \begin{equation}
        \mathcal R\left(\mathbf Z\left(\{\widetilde{\mathbf{F}}_s\}, \boldsymbol\rho\right)\right) \subseteq \mathcal R\left(\mathbf Z\left({\mathbf{F}}, \overline{\boldsymbol\rho}\right)\right). \label{theorem_eq}
    \end{equation}
    where $ \mathcal R$ is the range operator. 
\end{prop}
\begin{proof} Assume that \eqref{theorem_eq} holds, so that for any $\bm{\alpha} \in \mathbb{C}^{L \times 1}$, there exists $\bm{\lambda}$  satisfying,
%\begin{equation}
%
$\mathbf{Z}\left(\{\widetilde{\mathbf{F}}_s\}, \bm{\rho}\right) \bm{\alpha}  = \mathbf{Z}(\mathbf{F}, \overline{\bm{\rho}}) \bm{\lambda}$.
%\end{equation}
 Since $\mathbf{Z}(\mathbf{F}, \overline{\bm{\rho}}) \mathbf{Z}^\dagger(\mathbf{F}, \overline{\bm{\rho}})$  is the orthogonal projector onto this subspace, it follows that $\left(\mathbf{I}_{SM}^{} - \mathbf{Z}(\mathbf{F}, \overline{\bm{\rho}}) \mathbf{Z}^\dagger(\mathbf{F}, \overline{\bm{\rho}})\right) \mathbf{Z}(\{\widetilde{\mathbf{F}}_s\},\bm{\rho})\, \bm{\alpha} = \bm{0}$, $\forall  \bm{\alpha}$, 
     %\begin{align}
%\left(\mathbf{I}_{SM}^{} - \mathbf{Z}^{\mathrm{nom}}(\overline{\bm{\theta}},\overline{\bm{\phi}}) (\mathbf{Z}^{\mathrm{nom}}(\overline{\bm{\theta}},\overline{\bm{\phi}}))^\dagger\right) &\mathbf{Z}^{\mathrm{spoof}}(\bm{\theta},\bm{\phi})\, \bm{\alpha} = \bm{0}, \quad \forall  \bm{\alpha} \nonumber     \end{align}
    and hence $C(\overline{\bm{\rho}}) = 0$.
Conversely, if spoofing is perfect, then for any received signal $\mathbf{y}$ there exists $\bm{\lambda}$  such that $\mathbf{y} = \mathbf{Z}(\mathbf{F}, \overline{\bm{\rho}}) \bm{\lambda}$. Since this holds for any $\bm{\alpha}$, $\mathcal R\left(\mathbf Z\left(\{\widetilde{\mathbf{F}}_s\}, \boldsymbol\rho\right)\right) \subseteq \mathcal R(\mathbf Z(\mathbf{F}, \overline{\boldsymbol\rho}))$.
\end{proof}
\begin{remark}
 Existing pilot-based spoofing in \cite[Proposition 4]{italiano2025holotrace} aims to manipulate the received signal by designing transmit pilots $x_{s,m}$. However, it does not enforce that the received signal corresponding to the true channel lies in the subspace spanned by the spoofed channel parameters as indicated in \eqref{theorem_eq}. Unlike pilot design, precoder design provides more degrees of freedom and allows reshaping of the observation subspace.
\end{remark}
Since the true channel gain $\bm{\alpha}$ is unknown at the UE, we introduce auxiliary variables $\mathbf{d}$ and solve the following optimization problem:
\begin{align} \label{relaxed_P1}
\min_{\{\widetilde{\mathbf{F}}_s\}, \mathbf{d}, \bm{\lambda}} & \sum_{s=1}^S \big\|  \mathbf{Z}_s\left(\{\widetilde{\mathbf{F}}_s\},{\bm{\rho}}\right)\mathbf{d}  - \mathbf{Z}_s(\mathbf{F}, \overline{\bm{\rho}})\bm{\lambda}\big\|^2 \\ 
   \text{subject to}\quad &  |\widetilde{\mathbf{F}}_{s} (i,j)| \leq \frac{1}{\sqrt{N_t}}, \quad \forall i,j,s,  \\
   & \big\| \bm{\lambda} \big\|^2  \leq P_{\text{max}},
\end{align}
where $P_{\max}$ denotes the maximum allowable power of the equivalent channel gain vector $\boldsymbol{\lambda}$. The analog precoder with relaxed constraint can be achieved with phase shifters and attenuators. Each RF chain has both a phase shifter and a variable attenuator. The attenuator reduces the magnitude below ${1}/{\sqrt{N_t}}$ while the phase shifter controls the phase \cite{5749702}. Since the above problem is not a convex problem, we solve via block coordinate descent over $\{\widetilde{\mathbf{F}}_s\}, \mathbf{d}, \bm{\lambda}$. 
\subsection{Alternating Optimization}

For fixed $\mathbf d$ and $\bm\lambda$, each term depends only on the corresponding precoder $\tilde{\mathbf{F}}_s$ making the optimization problem separable across the measurements. The precoder design at the $s$-th measurement is
\begin{align}
    \min_{\widetilde{\mathbf F}_s}\quad
    &\Vert
    \widetilde{\mathbf F}_s^{\mathsf H}\mathbf b_s
    -
    \widetilde{\mathbf y}_s
    \Vert^2
    \label{eq:F_subproblem}
    \\
    \text{s.t.}\quad
    &|\widetilde{\mathbf F}_s(i,m)| \le c,
    \qquad \forall i,m,
    \nonumber
\end{align}
where $\mathbf b_s \triangleq  \mathbf{A}_{\mathrm{UE}}^\star  \text{diag}^{\mathsf H} (\mathbf{d}) \mathbf{A}_{\mathrm{BS}}^{\mathsf H} \mathbf w_s$, $ \widetilde{\mathbf y}_s
\triangleq
\mathbf Z_s(\mathbf{F},\overline{\bm\rho})\bm\lambda$, and $c \triangleq {1}/{\sqrt{N_t}}$.
%\[
%\mathbf b_s \triangleq \mathbf H^{\mathsf H}\mathbf w_s,
%\qquad
%\mathbf y_s^{\mathrm{desired}}
%\triangleq
%\mathbf Z_s^{\mathrm{nom}}(\overline{\bm\theta},\overline{\bm\phi})\bm\lambda,
%\qquad
%c \triangleq \frac{1}{\sqrt{N_t}}.
%\]
Problem \eqref{eq:F_subproblem} is convex and separable across the $M$ columns of $\widetilde{\mathbf F}_s$. Specifically, letting $\widetilde{\mathbf f}_{s,m}$ denote the $m$-th column of $\widetilde{\mathbf F}_s$, we obtain
\begin{equation}
    \min_{\widetilde{\mathbf f}_{s,m}}
    \vert
    \mathbf b_s^{\mathsf H}\widetilde{\mathbf f}_{s,m}
    -
    \widetilde{y}_{s,m}
    \vert^2
    \quad
    \text{s.t.}\quad
    |\widetilde{\mathbf f}_{s,m}(i)| \le c,\ \forall i.
    \label{eq:column_subproblem}
\end{equation}
For \eqref{eq:column_subproblem}, the optimal phase of each entry aligns the terms
$b_s^*(i)\widetilde{\mathbf f}_{s,m}(i)$ with the target phase
$\angle \widetilde{y}_{s,m}$. Hence, the optimizer admits the form $\widetilde{\mathbf f}_{s,m}(i)
    =
    \rho_{i,m}\,
    e^{j(\angle \widetilde{y}_{s,m}+\angle b_s(i))}$ 
Substituting this into \eqref{eq:column_subproblem} reduces the problem to
\begin{equation}
    \min_{\{ \rho_{i,m} \in [0,c]\}}
    \big(
    \sum_{i=1}^{N_t}|b_s(i)|\rho_{i,m}
    -
    |\widetilde{y}_{s,m}|
    \big)^2.
    %\quad
    %\text{s.t.}\quad
    %0\le \rho_{i,m}\le c.
    \label{eq:rho_subproblem}
\end{equation}
A convenient optimal solution is $\rho_{i,m}
    =
    \min\{c,\beta_m |b_s(i)|\}$, 
where $\beta_m\ge 0$ is chosen such that
\begin{equation}
    \sum_{i=1}^{N_t}
    |b_s(i)|\min\{c,\beta_m |b_s(i)|\}
    =
    \min\{|\widetilde{y}_{s,m}|,\ c\|\mathbf b_s\|_1\}.
    \label{eq:beta_eq}
\end{equation}
Therefore, the precoder update is
\begin{equation}
    \widetilde{\mathbf f}_{s,m}(i)
    =
    \min\{c,\beta_m |b_s(i)|\}
    e^{j(\angle \widetilde{y}_{s,m}+\angle b_s(i))}.
    \label{eq:fopt_corrected}
\end{equation}
The smallest value of $\beta_m$ can be found by bisection from \eqref{eq:beta_eq}.
For known $\{\widetilde{\mathbf{F}}_s\}, \bm{\lambda}$, the optimal value of $\mathbf{d}$ is found as
%\subsection{Design of  $\mathbf{d}$}
\begin{equation}
         \mathbf{d} = \mathbf{Z}^{\dagger}\left(\{\widetilde{\mathbf{F}}_s\}, \bm{\rho} \right)\mathbf{Z}\left({\mathbf{F}}, \overline{\bm{\rho}}\right) \bm{\lambda}.  \label{d}
\end{equation}
Similarly, for known $\{\widetilde{\mathbf{F}}_s\}, \mathbf{d}$, the optimal value of $\bm{\lambda}$ is 
%\subsection{Design of  \texorpdfstring{$\boldsymbol{\lambda}$}{lambda}}
\begin{align}
 &\bm{\lambda}= \label{lambda} \\
&\begin{cases}
 \mathbf{Z}^\dagger ({\mathbf{F}}, \overline{\bm{\rho}})  \mathbf{Z}\left(\{\widetilde{\mathbf{F}}_s\}, \bm{\rho}  \right) \mathbf{d} ,  \quad \text{if }\big\| \bm{\lambda} \big\|^2  \leq P_{\text{max}}, \\[6pt]
\left(\mathbf{Z}^\mathsf{H}({\mathbf{F}}, \overline{\bm{\rho}})\mathbf{Z}({\mathbf{F}}, \overline{\bm{\rho}})+ \mu \mathbf{I}\right)^{-1}\mathbf{Z}^{\mathsf{H}}({\mathbf{F}}, \overline{\bm{\rho}})\mathbf{Z}\left(\{\widetilde{\mathbf{F}}_s\}, \bm{\rho}  \right)\mathbf{d} , & \text{else},
\end{cases}\nonumber 
\end{align}
where $ \mu  \geq 0$ is chosen such that $\big\| \bm{\lambda} \big\|^2  = P_{\text{max}}$. 
The updates in \eqref{eq:fopt_corrected}, \eqref{d}, \eqref{lambda} are iterated till convergence.

\subsection{Convergence Analysis}

Defining the objective function in \eqref{relaxed_P1} as $C(\{\widetilde{\mathbf{F}}_s\}, \mathbf{d}, \boldsymbol{\lambda})$, and letting 
\( C^k \triangleq C(\{\widetilde{\mathbf{F}}_s\}^k, \mathbf{d}^k, \boldsymbol{\lambda}^k) \), 
at iteration \( k \), the updates of $\{\widetilde{\mathbf{F}}_s\}$, $\mathbf{d}$, and $\boldsymbol{\lambda}$ are obtained by solving their respective subproblems optimally. Therefore,
\begin{align}
C(\{\widetilde{\mathbf{F}}_s\}^{k+1}, \mathbf{d}^k, \boldsymbol{\lambda}^k)
&\le C(\{\widetilde{\mathbf{F}}_s\}^{k}, \mathbf{d}^k, \boldsymbol{\lambda}^k), \\
C(\{\widetilde{\mathbf{F}}_s\}^{k+1}, \mathbf{d}^{k+1}, \boldsymbol{\lambda}^k)
&\le C(\{\widetilde{\mathbf{F}}_s\}^{k+1}, \mathbf{d}^{k}, \boldsymbol{\lambda}^k), \\
C(\{\widetilde{\mathbf{F}}_s\}^{k+1}, \mathbf{d}^{k+1}, \boldsymbol{\lambda}^{k+1})
&\le C(\{\widetilde{\mathbf{F}}_s\}^{k+1}, \mathbf{d}^{k+1}, \boldsymbol{\lambda}^{k}).
\end{align}
Combining the above inequalities yields $C^{k+1} \le C^k.$ Since $C(\cdot) \ge 0$, the sequence $\{C^k\}$ is monotonically decreasing and lower bounded, and hence converges. 
\subsection{Computational Complexity}

We compare the computational complexity of the proposed precoder design with the pilot design in \cite{italiano2025holotrace}. The update of $\tilde{\mathbf{F}}_s$ in \eqref{eq:fopt_corrected} is separable across the $M$ columns. Each column update requires $\mathcal{O}(N_t)$ operations, leading to the total complexity is $\mathcal{O}(SM N_t).$ On the other hand, the pilot design in \cite[Eq. 32]{italiano2025holotrace} results in complexity $\mathcal{O}(SM L^2 + L^3).$

\section{Performance Evaluation}
\begin{figure}[t]
    \centering
    \begin{subfigure}[b]{0.42\textwidth}
        \centering
    \includegraphics[width=\linewidth]{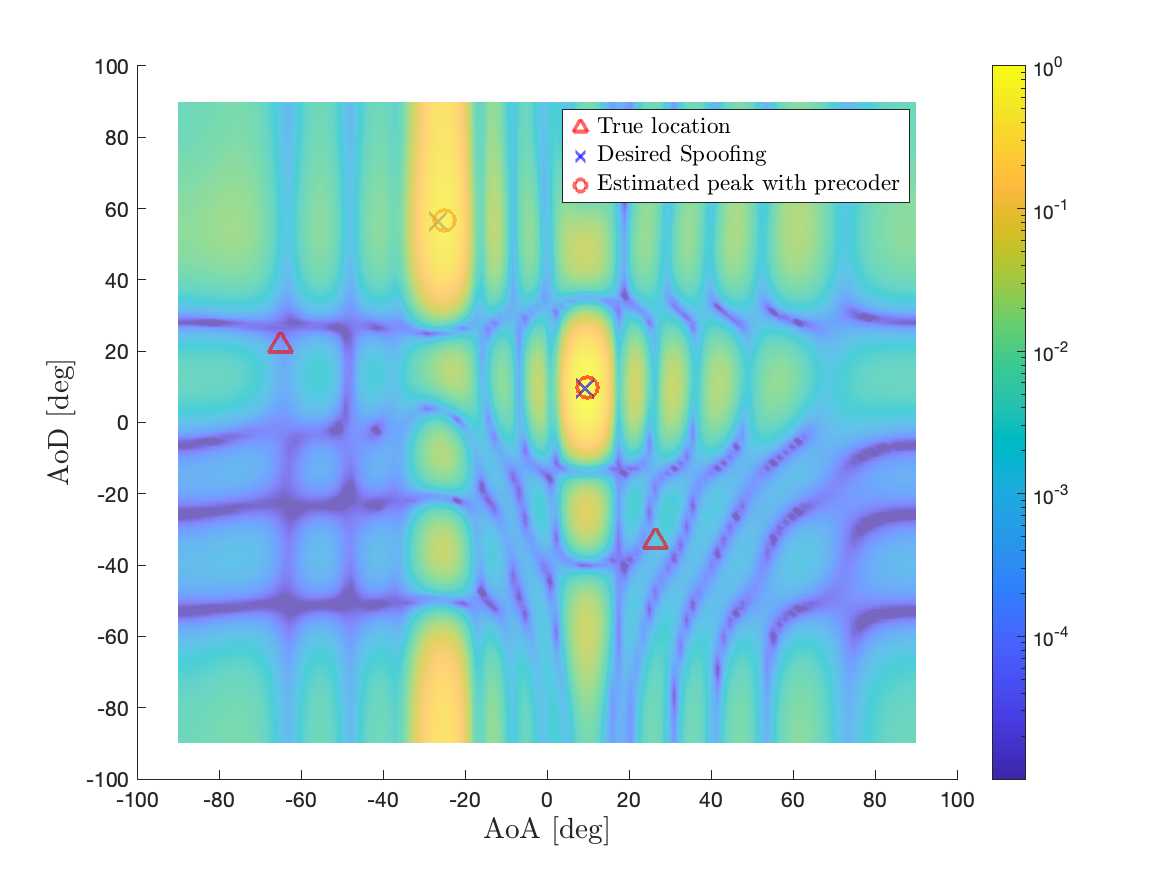}
        \caption{Precoder design (proposed)}
        \label{fig:precoder}
    \end{subfigure}
    \hfill
    %\vspace{-6mm}
    % Subfigure (b)
        \begin{subfigure}[b]{0.42\textwidth}
        \centering
        \includegraphics[width=\linewidth]{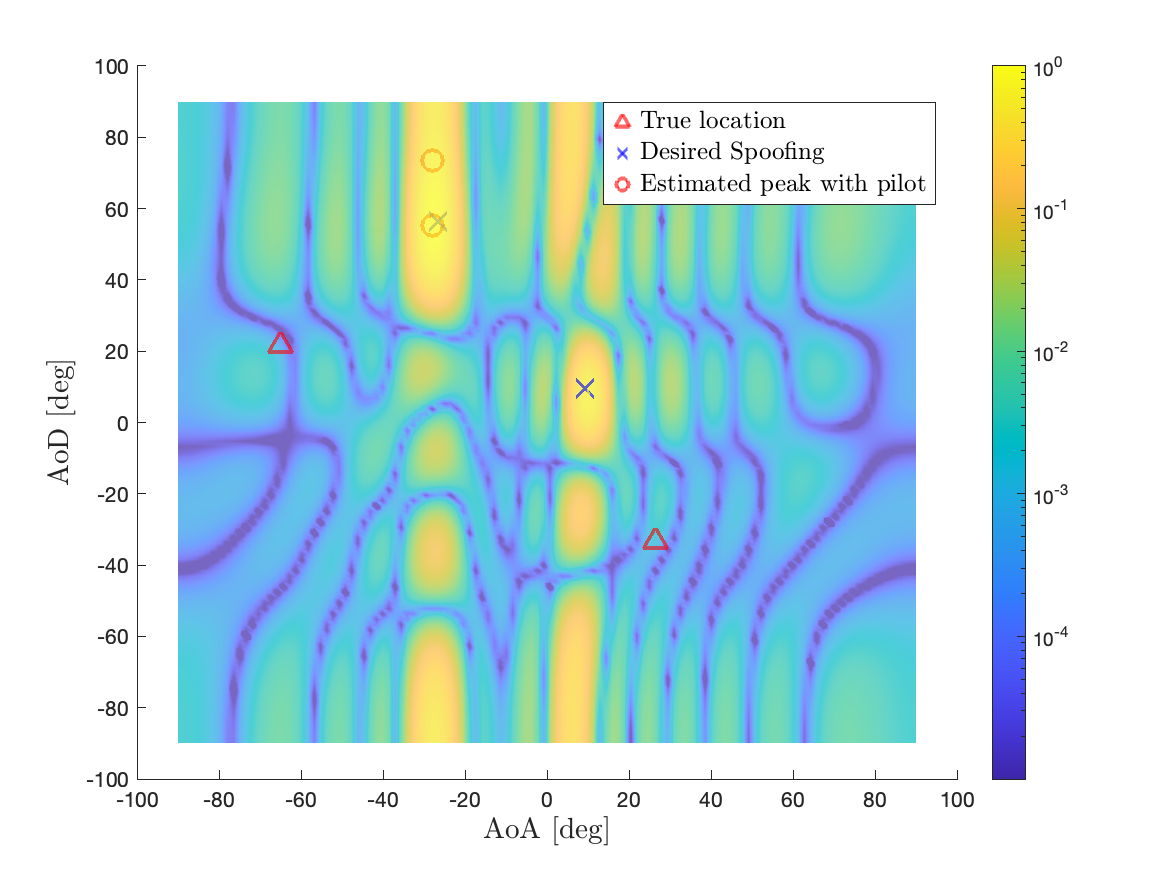}
        \caption{Pilot design \cite{italiano2025holotrace}}
        \label{fig:pilot}
    \end{subfigure}
    \caption{ Blind Spoofing (a) AoA/AoD attack with proposed analog precoder design, (b) AoA/AoD attack with pilot design \cite{italiano2025holotrace}.}
   \label{fig:heatmap}
\end{figure}
\subsection{Simulation Setup}
Using the \ac{bs} as a reference point for our local coordinates, we placed the \ac{bs} in 
\(\mathbf{p}_\mathrm{BS}^{}=\left[0 \; 0 \right]^\mathsf{T}\mathrm{m}\) with \(o_\mathrm{BS}^{}=0 \,\mathrm{rad}\), 
the \ac{ue} in \(\mathbf{p}_\mathrm{UE}^{}=\left[10\; 5\right]^\mathsf{T}\mathrm{m}\) with \(o_\mathrm{UE}^{}=\text{-}\frac{2}{3}\pi\,\mathrm{rad}\), 
and one \ac{sp} 
in \(\mathbf{p}_1^{}=\left[7\;\text{-}15\right]^\mathsf{T}\mathrm{m}\). 
The received signal in \eqref{eq:rx_signal} are modeled with the following parameters: $N_{r}^{}=15$, $N_{t}^{}=5$, $L=2$, $M=15$, $S=15$ with a bandwidth of $\mathrm{BW}=396$~MHz, and $f_c^{}=27.8$~GHz. With the given positions, the true measurements are $\bm{\theta}=\left[0.46\; \text{-}1.13\right]^\mathsf{T}\mathrm{rad}$, and $\bm{\phi}=\left[\text{-}0.58\; 0.37\right]^\mathsf{T}\mathrm{rad}$.
The spoofed position is chosen at \(\overline{\mathbf{p}}_\mathrm{UE}^{}=\left[30\; \text{}5\right]^\mathsf{T}\mathrm{m}\) with \(\overline{o}_\mathrm{UE}^{}=-\pi\,\mathrm{rad}\), and the \ac{sp} in \(\overline{\mathbf{p}}_1^{}=\left[20\; \text{-}10\right]^\mathsf{T}\mathrm{m}\). 
The respective spoofing measurements are $\overline{\bm{\theta}}=\left[0.59\; \text{-}0.46\right]^\mathsf{T}\mathrm{rad}$, $\overline{\bm{\phi}}=\left[0.59\; 1.25\right]^\mathsf{T}\mathrm{rad}$. We define spoofing deviation for AoA/ AoD as
    $\varepsilon_{\mathrm{dev}}^{\text{AoA}} = \Vert \widehat{\bm{\theta}} - \overline{\bm{\theta}}\Vert$ and $\varepsilon_{\mathrm{dev}}^{\text{AoD}} = \Vert \widehat{\bm{\phi}} - \overline{\bm{\phi}}\Vert$, respectively, 
    representing the deviation between the estimated and the {desired} spoofed angles at the BS. All results are expressed as \ac{rmse} values, computed over $T = 250$ independent Monte Carlo simulations. 
The \ac{rmse} is defined, for a generic error $\varepsilon$, as $\mathrm{RMSE}=\sqrt{{1}/{T}\sum^{T}_{t=1}\varepsilon_t^2}.$

\subsection{Simulation Results}
The simulation results highlight a significant performance gap between the proposed analog precoder design and the baseline pilot-based spoofing approach \cite{italiano2025holotrace}. As shown in Fig.~\ref{fig:heatmap}, UE’s capacity to manipulate the BS’s spatial perception is significantly more precise when utilizing precoder design. In Fig.~\ref{fig:precoder}, the proposed analog precoder design achieves near-perfect alignment between the desired location and the estimated peaks. The cost function heatmap obtained from \eqref{eq:spoofingCriterion1} displays sharp, well-defined energy concentrations at the target AoA and AoD coordinates. This indicates that the precoder successfully spoofs the spatial signature, even without knowledge of true channel gains, and multiple spoofed paths are present. Conversely, Fig.~\ref{fig:pilot} illustrates the limitations of a pilot-only spoofing strategy. The BS fails to resolve the intended spoofed angles correctly.

Fig.~\ref{fig:rmse} shows that precoder-based spoofing errors for both AoA $\varepsilon_{\mathrm{dev}}^{\mathrm{AoA}}$ and AoD $\varepsilon_{\mathrm{dev}}^{\mathrm{AoD}}$ continue to decay toward zero, demonstrating that the spoofing becomes increasingly precise as noise is reduced. By exploiting the antenna array spatial degrees of freedom, the precoder creates a synthetic steering vector that matches the target AoA/AoD. On the other hand, pilot-only spoofing provides insufficient control over the effective spatial signature, leading to an RMSE floor. 
\begin{figure}
    \centering
\definecolor{mycolor1}{rgb}{0.14902,0.14902,0.14902}%
\definecolor{mycolor2}{rgb}{1.00000,0.41176,0.16078}%
\begin{tikzpicture}

\begin{axis}[%
width=0.761\linewidth,
height=0.601\linewidth,
at={(0\linewidth,0\linewidth)},
scale only axis,
xmin=-10,
xmax=20,
xlabel style={font=\color{white!15!black}},
xlabel={Transmitted Power [dBm]},
ymode=log,
ymin=0.654711452078273,
ymax=100,
yminorticks=true,
ylabel style={font=\color{white!15!black}},
ylabel={RMSE [deg]},
axis background/.style={fill=white},
xmajorgrids,
ymajorgrids,
yminorgrids,
legend style={at={(0.178,0.119)}, anchor=south west, legend cell align=left, align=left, draw=white!15!black},
clip=false,
legend style={at={(0.02,0.02)}, anchor=south west, draw=black, fill=white, font=\footnotesize}
]
\addplot [color=mycolor1, line width=1.0pt, mark=triangle, mark options={solid, mycolor1}]
  table[row sep=crcr]{%
-10	36.3733256025921\\
-5	27.6342916915155\\
0	12.0831273227798\\
5	1.42608759005213\\
10	1.24756268366764\\
15	0.888240553841458\\
20	0.654711452078273\\
};
\addlegendentry{AoA with precoder design}

\addplot [color=mycolor2, line width=1.0pt, mark=triangle, mark options={solid, rotate=270, mycolor2}]
  table[row sep=crcr]{%
-10	37.0552431956653\\
-5	33.0552431956653\\
0	33.0552431956653\\
5	33.0552431956653\\
10	33.0552431956653\\
15	33.0552431956653\\
20	33.0552431956653\\
};
\addlegendentry{AoA with pilot design}

\addplot [color=blue, line width=1.0pt, mark=diamond, mark options={solid, blue}]
  table[row sep=crcr]{%
-10	49.0123657528771\\
-5	44.8170218984004\\
0	17.0985366831494\\
5	2.75949154919256\\
10	1.93456118731762\\
15	1.25473330700124\\
20	0.953820296605059\\
};
\addlegendentry{AoD with precoder design}

\addplot [color=red, line width=1.0pt, mark=triangle, mark options={solid, rotate=90, red}]
  table[row sep=crcr]{%
-10	70.2222187804389\\
-5	68.7777401553174\\
0	67.6371120106501\\
5	67.6364354632967\\
10	67.6364354632967\\
15	67.6364354632967\\
20	67.6364354632967\\
};
\addlegendentry{AoD with pilot design}

\end{axis}
\end{tikzpicture}%
    \caption{Spoofing deviation in AoA/AoD.}
    \label{fig:rmse}
\end{figure}
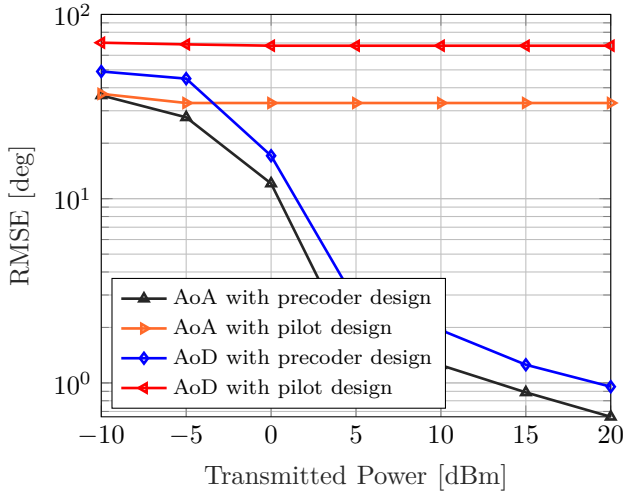
\begin{comment}
  Fig.~\ref{fig:spectral_eff} shows the achievable communication rate for two spoofing strategies versus transmit power. From an adversarial perspective, the objective is not only to mislead the base station about the user’s location, but also to maintain a high communication rate. The figure illustrates how this trade-off depends on the chosen virtual location. Precoder-based spoofing provides better angular control, but its rate depends strongly on the chosen virtual location. If the chosen spoofed location  $(\overline{\mathbf{p}}_\mathrm{UE}^{1}  = \left[30\; 5\right]^\mathsf{T}\mathrm{m}$) corresponds to low channel gain, the precoder-based method (Blue Solid) suffers a rate loss relative to the pilot-only baseline, since part of the spatial degrees of freedom must be used to enforce the false geometry. By contrast, when the spoofed location $(\overline{\mathbf{p}}_\mathrm{UE}^{2}  = \left[30\; 18.4\right]^\mathsf{T}\mathrm{m}$)  is aligned with the true channel conditions, the precoder-based attack (Blue Dashed) achieves a rate close to the ``No Spoof'' case. Pilot-based spoofing (Pilot \#1 and \#2) is less sensitive to the spoofed coordinates, since it does not actively reshape the spatial signature.  
\end{comment}

Fig.~\ref{fig:spectral_eff} shows the achievable communication rate for two spoofing strategies versus transmit power. From an adversarial perspective, the objective is to mislead the BS while maintaining a high communication rate, revealing a fundamental trade-off between spoofing effectiveness and communication performance. Precoder-based spoofing enables accurate angular control, but its rate is sensitive to the chosen virtual location. For a spoofed location with weak channel direction $(\overline{\mathbf{p}}_\mathrm{UE}^{1}  = [30 \; 5]^\mathsf{T}\mathrm{m})$, the rate decreases (Blue Solid) since part of the spatial degrees of freedom must be used to enforce the false geometry. In contrast, when the spoofed location aligns with favorable channel conditions $(\overline{\mathbf{p}}_\mathrm{UE}^{2} = [30  \;18.4]^\mathsf{T}\mathrm{m})$, the rate approaches the ``No Spoof'' case (Blue Dashed). Pilot-based spoofing (Pilot \#1 and \#2) is less sensitive to the spoofed location, as it does not actively reshape the spatial signature.

%Fig.~\ref{fig:spectral_eff} illustrates the achievable communication rate under two distinct spoofing strategies across varying transmit power levels. The results demonstrate that while precoder-based spoofing offers superior angular accuracy, its impact on the communication link is highly dependent on the target virtual location. When the spoofed location lies in a region of inherently low channel gain, the precoder-based approach (Blue Triangle) yields a lower rate compared to the pilot-only baseline. This occurs because the analog precoder must allocate significant spatial resources to override the physical channel geometry, resulting in a beamforming penalty that reduces overall signal strength. When the spoofed position is optimized to match the channel characteristics of the true location, the precoder-based attack (Blue Solid) achieves a rate nearly identical to the "No Spoof" scenario. In this case, the spatial signature is synthesized efficiently, allowing the attacker to maintain high-speed communication while simultaneously deceiving the BS’s positioning system. In contrast, the communication rate for pilot-based spoofing (Pilot $\#1$ and $\#2$ ) remains relatively stable and less sensitive to the chosen spoofing coordinates. Since the pilot-only attack lacks beamforming control and does not attempt to reshape the physical wavefront, it avoids the heavy beamforming penalty of a poorly placed precoder.

\begin{figure}
    \centering
\definecolor{mycolor1}{rgb}{0.85000,0.32500,0.09800}%
\definecolor{mycolor2}{rgb}{0.92900,0.69400,0.12500}%
\begin{tikzpicture}

\begin{axis}[%
width=0.761\linewidth,
height=0.6\linewidth,
at={(0\linewidth,0\linewidth)},
scale only axis,
xmin=-15,
xmax=45,
xlabel style={font=\color{white!15!black}},
xlabel={Transmitted Power [dBm]},
ymode=log,
ymin=7.55673137196853,
ymax=39.9857810764912,
yminorticks=true,
ylabel style={font=\color{white!15!black}},
ylabel={Spectral efficiency [bps/Hz]},
axis background/.style={fill=white},
xmajorgrids,
ymajorgrids,
yminorgrids,
legend style={at={(0.659,0.106)}, anchor=south west, legend cell align=left, align=left, draw=white!15!black},
clip=false,
legend style={at={(0.98,0.02)}, anchor=south east, draw=black, fill=white, font=\footnotesize}
]
\addplot [color=blue, line width=2.0pt, mark=triangle, mark options={solid, rotate=270, blue}]
  table[row sep=crcr]{%
-15	7.55673137196853\\
-10	8.74720613496729\\
-5	10.0167681216537\\
0	11.3592506364392\\
5	12.758118957161\\
10	14.1960256456239\\
15	15.6609543698227\\
20	17.1457330898981\\
25	18.6470089161029\\
30	20.1646635534224\\
35	21.7001264670523\\
40	23.256239763799\\
45	24.8370825782843\\
};
\addlegendentry{Spoof with precoder \#1}

\addplot [color=red, line width=2.0pt, mark=diamond, mark options={solid, red}]
  table[row sep=crcr]{%
-15	17.9395123030714\\
-10	19.6004724262776\\
-5	21.2614352327662\\
0	22.9223988877853\\
5	24.5833628111333\\
10	26.2443268193345\\
15	27.9052908543687\\
20	29.5662548978881\\
25	31.2272189440908\\
30	32.8881829911421\\
35	34.5491470384617\\
40	36.2101110858661\\
45	37.8710751332974\\
};
\addlegendentry{Spoof with pilot \#1}

\addplot [color=blue, dashed, line width=2.0pt, mark=triangle, mark options={solid, rotate=270, blue}]
  table[row sep=crcr]{%
-15	19.024167490048\\
-10	20.5961453940658\\
-5	22.185745145941\\
0	23.7891096566338\\
5	25.4041955291474\\
10	27.0304485370434\\
15	28.6671653945391\\
20	30.3121178504158\\
25	31.9629124639646\\
30	33.6181594715515\\
35	35.2765339905944\\
40	36.9365308587044\\
45	38.5971694127473\\
};
\addlegendentry{Spoof with precoder \#2}

\addplot [color=mycolor1, dashed, line width=2.0pt, mark=diamond, mark options={solid, mycolor1}]
  table[row sep=crcr]{%
-15	15.7580379991992\\
-10	17.418984245554\\
-5	19.0799426637534\\
0	20.7409049310692\\
5	22.4018684155865\\
10	24.0628322850172\\
15	25.7237962761683\\
20	27.3847603058106\\
25	29.0457243476251\\
30	30.7066883932886\\
35	32.3676524401694\\
40	34.028616487435\\
45	35.6895805348224\\
};
\addlegendentry{Spoof with pilot \#2}

\addplot [color=mycolor2, line width=2.0pt, mark=o, mark options={solid, mycolor2}]
  table[row sep=crcr]{%
-15	20.0542138322846\\
-10	21.7151769736489\\
-5	23.376140734565\\
0	25.0371046914007\\
5	26.6980687101916\\
10	28.3590327485745\\
15	30.0199967931529\\
20	31.6809608396905\\
25	33.3419248868477\\
30	35.0028889342007\\
35	36.6638529816158\\
40	38.3248170290504\\
45	39.9857810764912\\
};
\addlegendentry{No spoof}

\end{axis}
\end{tikzpicture}%
    \caption{Spectral efficiency performance comparison.}
    \label{fig:spectral_eff}
\end{figure}
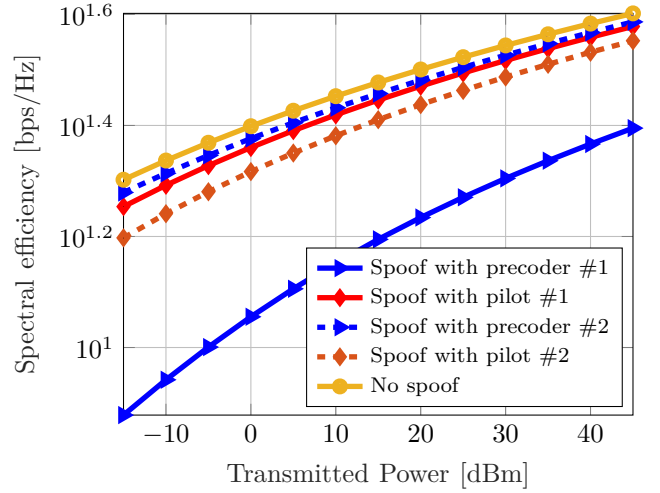
\section{Conclusion}
This paper considered the problem of precoder design for location privacy against a single BS performing multipath-based localization in a MIMO uplink setting, without relying on CSI.  We show that geometry-aware \ac{csi}-blind precoder design provides substantially stronger spoofing capability than pilot-only approaches. The additional spatial degrees of freedom provided by the analog precoder enable structured and controllable position deception in multipath environments. However, improved spoofing performance comes at a communication cost. The results show a clear trade-off between privacy and link quality. Therefore, the effectiveness of blind spatial spoofing should be evaluated jointly in terms of localization deception and communication cost. Future work includes extending the design to multicarrier systems, as well as to incorporate more realistic constraints, imperfect hardware calibration, and dynamic channels. Finally, extending the threat model to multi-BS localization and studying robust countermeasures against blind precoder-based spoofing would provide a more complete understanding of privacy vulnerabilities in future wireless localization systems.

\balance
\bibliographystyle{IEEEtran}
\bibliography{sub/ref}

@IEEEtranBSTCTL{IEEEexample:BSTcontrol,
  CTLuse_article_number     = "yes",
  CTLuse_paper              = "yes",
  CTLuse_url                = "no",
  CTLuse_forced_etal        = "yes",
  CTLmax_names_forced_etal  = "1",
  CTLnames_show_etal        = "1",
  CTLuse_alt_spacing        = "yes",
  CTLalt_stretch_factor     = "10",
  CTLdash_repeated_names    = "yes",
  CTLname_format_string     = "{f.~}{vv~}{ll}{, jj}",
  CTLname_latex_cmd         = "",
  CTLname_url_prefix        = "[Online]. Avafilable:"
}

@article{zhang_privacy_2025,
	title = {Privacy Preservation in {MIMO}-{OFDM} Localization Systems: A Beamforming Approach},
	issn = {2162-2337, 2162-2345},
	shorttitle = {Privacy {Preservation} in {MIMO}-{OFDM} {Localization} {Systems}},
	doi = {10.1109/LWC.2025.3560219},
	journal = {IEEE Wireless Communications Letters},
	author = {Zhang, Yuchen and Chen, Hui and Keskin, Musa Furkan and Pourafzal, Alireza and Zheng, Pinjun and Wymeersch, Henk and Al-Naffouri, Tareq Y.},
	year = {2025},
    volume={14},
    number={7},
    pages={1979-1983},
}

@inproceedings{GOSPA,
	address = {Xi'an, China},
	title = {Generalized optimal sub-pattern assignment metric},
	isbn = {9780996452700},
	url = {http://ieeexplore.ieee.org/document/8009645/},
	doi = {10.23919/ICIF.2017.8009645},
	urldate = {2025-06-13},
	booktitle = {2017 20th {International} {Conference} on {Information} {Fusion} ({Fusion})},
	publisher = {IEEE},
	author = {Rahmathullah, Abu Sajana and Garcia-Fernandez, Angel F. and Svensson, Lennart},
	month = jul,
	year = {2017},
	pages = {1--8},
}

@article{checa2020location,
  title={Location-privacy-preserving technique for {5G} mmWave devices},
  author={Checa, Javier Jim{\'e}nez and Tomasin, Stefano},
  journal={IEEE Communications Letters},
  volume={24},
  number={12},
  pages={2692--2695},
  year={2020},
  publisher={IEEE}
}

@inproceedings{tomasin2022beamforming,
  title={Beamforming and artificial noise for cross-layer location privacy of e-health cellular devices},
  author={Tomasin, Stefano},
  booktitle={2022 IEEE International Conference on Communications Workshops (ICC Workshops)},
  pages={568--573},
  year={2022},
  organization={IEEE}
}

@article{li2024channel,
  title={Channel state information-free location-privacy enhancement: Fake path injection},
  author={Li, Jianxiu and Mitra, Urbashi},
  journal={IEEE Transactions on Signal Processing},
  year={2024},
pages={3745 - 3760},
volume={72},
  publisher={IEEE}
}

@article{li2025delay,
  title={Delay-Angle Information Spoofing for Channel State Information-Free Location-Privacy Enhancement},
  author={Li, Jianxiu and Mitra, Urbashi},
  journal={arXiv preprint arXiv:2504.14780},
  year={2025}
}

@ARTICLE{italiano2025tutorial,
  author={Italiano, Lorenzo and Camajori Tedeschini, Bernardo and Brambilla, Mattia and Huang, Huiping and Nicoli, Monica and Wymeersch, Henk},
  journal={IEEE Communications Surveys \& Tutorials}, 
  title={A Tutorial on {5G} Positioning}, 
  year={2025},
  volume={27},
  number={3},
  pages={1488-1535},
  doi={10.1109/COMST.2024.3449031}}

@ARTICLE{beamforming,
  author={Yi, Wang and Zhiqing, Wei and Zhiyong, Feng},
  journal={China Communications}, 
  title={Beam training and tracking in mmWave communication: A survey}, 
  year={2024},
  volume={21},
  number={6},
  pages={1-22},
  doi={10.23919/JCC.ea.2021-0873.202401}}

@ARTICLE{5749702,
  author={Gholam, Fouad and Via, Javier and Santamaria, Ignacio},
  journal={IEEE Transactions on Vehicular Technology}, 
  title={Beamforming Design for Simplified Analog Antenna Combining Architectures}, 
  year={2011},
  volume={60},
  number={5},
  pages={2373-2378},
  doi={10.1109/TVT.2011.2142205}}

@article{italiano2025holotrace,
  title={HoloTrace: a Location Privacy Preservation Solution for {mmWave MIMO-OFDM} Systems},
  author={Italiano, Lorenzo and Pourafzal, Alireza and Chen, Hui and Brambilla, Mattia and Seco-Granados, Gonzalo and Nicoli, Monica and Wymeersch, Henk},
  journal={arXiv preprint arXiv:2509.23444},
  doi = {10.48550/arXiv.2509.23444},
  year={2025}
}

@ARTICLE{6847111,
  author={Alkhateeb, Ahmed and El Ayach, Omar and Leus, Geert and Heath, Robert W.},
  journal={IEEE Journal of Selected Topics in Signal Processing}, 
  title={Channel Estimation and Hybrid Precoding for Millimeter Wave Cellular Systems}, 
  year={2014},
  volume={8},
  number={5},
  pages={831-846},
  doi={10.1109/JSTSP.2014.2334278}}

@article{witrisal2016high,
  title={High-accuracy localization for assisted living: {5G} systems will turn multipath channels from foe to friend},
  author={Witrisal, Klaus and Meissner, Paul and Leitinger, Erik and Shen, Yuan and Gustafson, Carl and Tufvesson, Fredrik and Haneda, Katsuyuki and Dardari, Davide and Molisch, Andreas F and Conti, Andrea and others},
  journal={IEEE Signal Processing Magazine},
  volume={33},
  number={2},
  pages={59--70},
  year={2016},
  publisher={IEEE}
}

@article{wymeersch20185g,
  title={{5G mmWave} positioning for vehicular networks},
  author={Wymeersch, Henk and Seco-Granados, Gonzalo and Destino, Giuseppe and Dardari, Davide and Tufvesson, Fredrik},
  journal={IEEE Wireless Communications},
  volume={24},
  number={6},
  pages={80--86},
  year={2018},
  publisher={IEEE}
}

@ARTICLE{Shahmansoori2018,
  author={Shahmansoori, Arash and Garcia, Gabriel E. and Destino, Giuseppe and Seco-Granados, Gonzalo and Wymeersch, Henk},
  journal={IEEE Transactions on Wireless Communications}, 
  title={Position and Orientation Estimation Through Millimeter-Wave {MIMO} in {5G} Systems}, 
  year={2018},
  volume={17},
  number={3},
  pages={1822-1835},
  doi={10.1109/TWC.2017.2785788}}

@ARTICLE{Sun2022,
  author={Sun, Yifei and Li, Jie and Zhang, Tong and Wang, Rui and Peng, Xiaohui and Han, Xiao and Tan, Haisheng},
  journal={Journal of Communications and Information Networks}, 
  title={An Indoor Environment Sensing and Localization System via {mmWave} Phased Array}, 
  year={2022},
  volume={7},
  number={4},
  pages={383-393},
  doi={10.23919/JCIN.2022.10005216}}

@INPROCEEDINGS{Lin2020,
  author={Lin, Zhipeng and Lv, Tiejun and Zhang, J. Andrew and Liu, Ren Ping},
  booktitle={2019 IEEE Global Communications Conference (GLOBECOM)}, 
  title={{3D} Wideband {mmWave} Localization for {5G} Massive {MIMO} Systems}, 
  year={2019},
  volume={},
  number={},
  pages={1-6},
  doi={10.1109/GLOBECOM38437.2019.9013639}}

@inproceedings{ardagna2007location,
  title={Location privacy protection through obfuscation-based techniques},
  author={Ardagna, Claudio A and Cremonini, Marco and Damiani, Ernesto and De Capitani di Vimercati, S and Samarati, Pierangela},
  booktitle={IFIP annual conference on data and applications security and privacy},
  pages={47--60},
  year={2007},
  organization={Springer}
}

@inproceedings{abedi2022non,
  title={Non-cooperative {Wi-Fi} localization \& its privacy implications},
  author={Abedi, Ali and Vasisht, Deepak},
  booktitle={Proceedings of the 28th Annual International Conference On Mobile Computing And Networking},
  pages={570--582},
  year={2022}
}

@inproceedings{zhang2024privacy,
  title={Privacy preservation in delay-based localization systems: Artificial noise or artificial multipath?},
  author={Zhang, Yuchen and Chen, Hui and Wymeersch, Henk},
  booktitle={GLOBECOM 2024-2024 IEEE Global Communications Conference},
  pages={2755--2760},
  year={2024},
  organization={IEEE}
}

@article{khan2025beamforming,
  title={On Beamforming for Transmitter Location Privacy in {MIMO} Systems},
  author={Khan, Umair Ali and Ho, Lester and Claussen, Holger and Flanagan, Mark F and Kundu, Chinmoy},
  journal={arXiv preprint arXiv:2508.09882},
  year={2025}
}
\end{document}